\documentclass[english,thm-restate,notab]{lipics-v2021}

\hideLIPIcs
\nolinenumbers

\usepackage{amsmath,amssymb,amsfonts}

\usepackage{verbatim}
\usepackage{xspace}
\usepackage{edtable}

\usepackage{dsfont}

\usepackage{hyperref}
\usepackage{xcolor}
\definecolor{darkred}{RGB}{220,50,0}
\definecolor{lightblue}{rgb}{.80,.85,1}
\definecolor{darkgreen}{RGB}{0,100,0}
\definecolor{firebrick}{RGB}{178,34,34}
\definecolor{salmon}{RGB}{250,128,114}
\definecolor{turquoise}{RGB}{0,128,114}
\definecolor{turquoise2}{RGB}{0,180,140}
\definecolor{darkorchid}{rgb}{0.60,0.20,0.80}
\definecolor{pink}{rgb}{1.,0.,0.6}
\definecolor{pink2}{rgb}{1.,0.,0.8}

\usepackage[normalem]{ulem}

\graphicspath{{pictures/}}

\theoremstyle{plain}

\newtheorem{assumption}[theorem]{\textbf{Assumption}}

\usepackage[section]{placeins}

\DeclareMathOperator{\ax}{{ax}}

\newcommand\R{\mathbb{R}}

\newcommand{\M}{\mathcal{M}}

\newcommand{\Su}{{\mathcal{S}}}

\newcommand{\lfs}{\operatorname{lfs}}
\newcommand{\Tan}{\operatorname{Tan}}
\newcommand{\Nor}{\operatorname{Nor}}

\newcommand{\BP}{\operatorname{BP}}
\newcommand{\UBP}{\operatorname{UBP}}
\newcommand{\rch}{\operatorname{rch}}

\newcommand{\bigO}{\mathcal{O}} 
\newcommand{\Lip}{\operatorname{Lip}}

\newcommand\abs[1]{\left\lvert#1\right\rvert}
\newcommand\norm[1]{\left\lVert #1 \right\rVert}

\title{
	The medial axis of any closed bounded set is Lipschitz stable with respect to the Hausdorff distance under ambient diffeomorphisms
}
\titlerunning{The medial axis of any closed set is Lipschitz stable
}

\author{Hana Dal Poz Kou\v{r}imsk\'a}{IST Austria \\{[Klosterneuburg, Austria]}}{hana.kourimska@ist.ac.at}{https://orcid.org/0000-0001-7841-0091}{}

\author{Andr{\'e} Lieutier}{No affiliation}{andre.lieutier@gmail.com }{}{}

\author{
	Mathijs Wintraecken}{Inria Sophia Antipolis, Universit{\'e} C{\^o}te d'Azur\\{[Sophia Antipolis, France]}  }{m.h.m.j.wintraecken@gmail.com}{https://orcid.org/0000-0002-7472-2220}{Supported by the European Union's Horizon 2020 research and innovation programme under the Marie Sk{\l}odowska-Curie grant agreement No. 754411, the Austrian science fund (FWF) grant No. M-3073, and the welcome package from IDEX of the Universit{\'e} C{\^o} d'Azur. }

\funding{This research has been supported by the European Research Council (ERC), grant No.\ 788183, by the Wittgenstein Prize, Austrian Science Fund (FWF), grant No.\ Z 342-N31, and by the DFG Collaborative Research Center TRR 109, Austrian Science Fund (FWF), grant No.\ I 02979-N35.
}

\acknowledgements{
	We are greatly indebted to Fred Chazal for sharing his insights.
	We further thank Erin Chambers, Christopher Fillmore, and Elizabeth Stephenson for early discussions and all members of the Edelsbrunner group (Institute of Science and Technology Austria) and the Datashape team (Inria) for the atmosphere in which this research was conducted. 
}

\authorrunning{
} 

\Copyright{
	Dominique Attali, Hana Dal Poz Kou\v{r}imsk\'a, Christopher Fillmore, Ishika Ghosh, Andr{\'e} Lieutier, Elizabeth Stephenson, and Mathijs Wintraecken
} 



\begin{document}
	\maketitle
	
	\begin{abstract} We prove that the medial axis of closed sets is Hausdorff stable in the following sense: 
		Let $\Su \subseteq \mathbb{R}^d$ be a fixed closed set that contains a bounding sphere.  That is, the bounding sphere is part of the set $\Su$. Consider the space of $C^{1,1}$~diffeomorphisms of~$\mathbb{R}^d$ to itself, which keep the bounding sphere invariant. 
		The map from this space of diffeomorphisms (endowed with a Banach norm) to the space of closed subsets of $\R^d$ (endowed with the Hausdorff distance), mapping a diffeomorphism $F$ to the closure of the medial axis of $F(\Su)$, is Lipschitz.
		This extends a previous stability result of Chazal and Soufflet on the stability of the medial axis of $C^2$~manifolds under $C^2$ ambient diffeomorphisms.  
	\end{abstract}

	\section{Introduction}
	
	
	In \cite{Federer}, Federer introduced the \emph{reach} of a (closed) set $\Su \subset \mathbb{R}^d$ as the infimum {over all points in $\Su$} of the distance from these points 
	to the \emph{medial axis} $\ax(\Su)$, the set of points in $\mathbb{R}^d$ for which the closest point in $\Su$ is not unique. 
		Federer also introduced the reach at a point $p\in\Su$ to be the distance from $p$ to the medial axis of $\Su$. We now call this quantity the \emph{local feature size} \cite{Amenta1999} and denote it by $\lfs(p)$. 
		

		Federer proved that the reach is stable under $C^{1,1}$~diffeomorphisms of the ambient space. Here, a $C^{1,1}$~map is a $C^1$~map whose derivative is Lipschitz, and a $C^{1,1}$~diffeomorphism is a $C^{1,1}$~bijective map whose inverse is also $C^{1,1}$. 
		Chazal and Soufflet~\cite{Chazal2004} proved that the medial axis is stable with respect to the Hausdorff distance under ambient diffeomorphisms, but under stronger assumptions than the work of Federer, namely assuming that $\Su$ is a $C^2$~manifold and the distortion is a $C^2$~diffeomorphism of the ambient space. 
		Chazal and Soufflet based their work on earlier results by Blaschke~\cite{Blaschke1916}, which were not as strong as Federer's.

		In this paper we extend the stability result of the medial axis. More concretely, we generalize the result of Chazal and Soufflet \cite{Chazal2004} to arbitrary closed sets and $C^{1,1}$~diffeomorphisms of the ambient space; we show that the Hausdorff distance between the medial axes of the closed set and its image is bounded in terms of Lipschitz constants stemming from the diffeomorphism of the ambient space. 
		{To be more concrete, we show (see Theorem \ref{ambientStabilityWithLipNorm}) that $d_H (\ax(\Su), \ax ( F(\Su ) )) 
			{= \bigO\left( r^2 \varepsilon \right) }$, where $r$ is the radius of the bounding sphere and $\varepsilon$ bounds how close the diffeomorphism is to the identity. }
		Our result follows from the work of Federer \cite{Federer} and in fact shortens the proof in \cite{Chazal2004} significantly.

		Our bounds on the Hausdorff distance say nothing about the topology of the medial axis, which is known to be highly unstable (see e.g. \cite{Attali2009}), although it preserves the homotopy type (see \cite{LIEUTIERhomotopytype}). 
		
		\begin{figure}[h!]
			\centering
			\includegraphics[width=.65\textwidth]{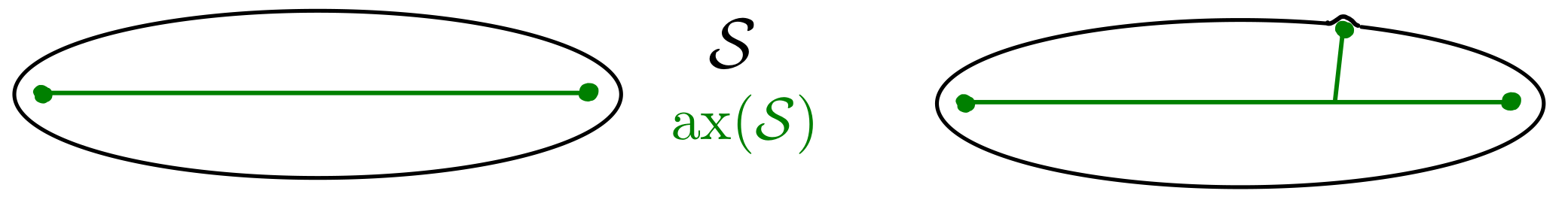}
			\caption{\small {Small (in terms of the Hausdorff distance) bumps on an even smooth curve or surface can create large new branches of the medial axis. However, these bumps in are not small in the $C^{2}$ or $C^{1,1}$ sense (roughly speaking the second order derivatives will be large) and therefore cannot be the consequence of an ambient diffeomorphism whose distortion is small in a $C^{2}$ or $C^{1,1}$ sense.}
			}
			\label{fig:bump}
		\end{figure}
		
		\subparagraph{Contribution and related work } 
	Our work differs from the majority of the literature in three essential ways:
	
	Firstly, we make no assumptions on the set we consider apart from that it is closed. 
	The stability 
	of the medial axis of (piecewise) smooth manifolds 
	has been the object of intense study, see for example \cite{Chazal2004, mather1983distance, thom1972, van2007maxwell, wolter1993cut, damon2021rigidity, Damon2008, damon2017medial, gasparovic2012blum, Wall1977}. 
	However, the manifold assumption is impossible to achieve in many applications --- such as in the context of astrophysics, one of the main motivations of this paper.
	
	Secondly, we achieve stability without pruning the medial axis. This contrasts with a large body of work, such as \cite{attali1997computing, cl2005lambda, damon2021rigidity, lieutier2023hausdorff}. 
	Not having to prune the medial axis is a significant advantage. On the downside, we limit the changes of the considered set to those induced by ambient diffeomorphisms. Nevertheless, given the standard examples of the instability of the medial axis --- see for example \cite{Attali2009} --- we believe these limitations are near to the weakest assumptions necessary for Hausdorff stability. 
	Within the context of ambient homeomorphisms, the results we obtain are close to optimal, as we specify in Remark \ref{rem:Opt}.  
	
	
	Thirdly, our results hold for sets in arbitrary dimensions and are not sensitive to the dimension of the set itself. A large part of the related work only investigates sets of low dimensions or codimension one manifolds, although there are some notable exceptions such as \cite{Wall1977}, see also \cite{damon2017medial}, and \cite{cl2005lambda,lieutier2023hausdorff}. 
	
	{\subparagraph{Motivation} 
		
		The medial axis has many real world applications --- among others, in robot motion planning \cite{latombe2012robot}, triangulation algorithms \cite{amenta2001power}, graphics \cite{tagliasacchi20163d}, vision \cite{hung2012medial, lescroart2013cortical}, and shape recognition, segmentation, and learning \cite{chambers2018medial,ho1986shape, SHAKED1998156, demir2019skelneton, Erin2016, trinh2011skeleton, hu2019mat,rezanejad2019scene}. See also the overviews \cite{tagliasacchi20163d, saha2016survey}.
		%
		The reach --- the distance between a set and its medial axis --- 
		is a central concept in manifold learning \cite{EddieManSTOC, aamari2018stability, fefferman2018fitting, fefferman2019fitting, fefferman2020reconstruction, sober2020manifold}.

		The motivation of this paper 
		is twofold:
		Firstly, we tackle the following challenge from
		the processing of images collected with optical devices which use lenses  --- such as cameras or telescopes:
		A shape extracted from such an image may be imprecise due to the imperfection of the lenses.

		Our result implies that the medial axis of such a shape is stable under these  imperfections. As a consequence, the outcome of any shape recognition or shape segmentation algorithm based on the medial axis will be stable.
		
		In addition to the disciplines listed at the beginning of this paragraph, the stability of the medial axis is sought after in astrophysics, in particular for shape analysis and automated shape identification in observational astronomy.
		%
		%
		Observational astronomers are interested in reconstructing objects like stars or galaxies, and their place in the universe from data gathered by telescopes. 
		They can deduce the distance from the object to the observer 
		thanks to so-called standard candles or red shift \cite{Fernie_1969, peebles1993principles, BigIdeasCosmology}. However, the image gets distorted due to optical effects --- either through gravitational lensing (\cite{Bartelmann2010}) or lensing inside the telescope itself (\cite{tang2017precision}).

		Such a distortion can be modeled as a diffeomorphism of the ambient space. At the same time, this problem cannot be tackled using the result by Chazal and Soufflet \cite{Chazal2004}, since the observed objects might not be smooth --- for example due to interactions with shock waves or jets. 
		In addition, with our method astrophysicists can not only reconstruct objects in space (3D), but also in spacetime (4D).

		{Another context where the removal of the assumption that the set is a (smooth) manifold is important is biology, as branching structures are ubiquitous in nature. In fact, it was questions from biology that motivated the `introduction'\footnote{The medial axis was studied before by Erd\H{o}s \cite{erdos1945some,erdos1946} in a different context. } of the medial axis by Blum~\cite{blum1967transformation}. 
		}

	}
	
	The second motivation is more formal in nature: The stability of the medial axis is instrumental in establishing its computability. 
	Indeed, when proving properties of algorithms based on the medial axis, authors generally assume the real RAM model.\footnote{The real RAM model is a standard, albeit non-realistic, assumption in Computational Geometry. It assumes one can calculate precisely with real numbers, instead of using $0$s and $1$s (which is the usual assumption in computer science). } 
	However, as was recently argued in \cite{lieutier2023hausdorff}, the medial axis needs to be stable in order to be computable in more realistic models of computation. 
	
	There is a more practical component to this formal question: It is not a priori clear if using possibly noisy real world data or the output of other computer programs as input for these algorithms yields answers that are close to the ground truth. 
	To be able to prove that the output is correct, we need (numerical) stability of the medial axis.  

	\subparagraph{Outline} 
	After revisiting preliminaries and known results in Section~\ref{sec:preliminaries},
	we state the main stability result in Section~\ref{AmbientDiff}. In Section~\ref{sec:BanachHolder} we reformulate this result in terms of norms on Banach spaces. This also exhibits the fact that the stability {of the medial axis} is 
	Lipschitz in the following sense: We think of the set $\Su$ as fixed and consider the map from the space of diffeomorphisms (endowed with a norm which makes it a Banach space) to {the space of closed subsets of $\R^d$} (endowed with the Hausdorff distance), {mapping each diffeomorphism $F:\R^d\to\R^d$ to the closure of the medial axis of $F(\Su)$}. The Lipschitz constant then only depends on the diameter of the bounding sphere of the set $\Su$.
		We conclude with some future work. 
	
	\section{Preliminaries: Sets of positive reach and the closest point projection}\label{sec:preliminaries}
	In this section we recall some definitions and results concerning {the medial axis and} sets of positive reach. 
	Essentially, we need three ingredients from the literature to prove our main theorem: the notions related to the closest point projection, the properties of the generalized normal and tangent spaces, and Federer's result on the stability of the reach under ambient diffeomorphisms \cite{Federer}.

	We write $d(\cdot,\cdot)$ for the Euclidean distance between two points, and the distance between a point and a set. That is, for any closed set $\Su$ and point $p$,
	\[
	d(p, \Su) = \inf_{q\in \Su} d(p,q).
	\]
	We denote the Hausdorff distance between two sets $A,B\subseteq \R^d$ by $d_H(A,B)$:
	\[
	d_H(A,B) = \max\left\{ \sup_{a\in A}d(a, B), \sup_{b\in B}d(b, A) \right\}.
	\]
	{We write $B(c,r)$, resp. $S(c,r)$, to denote balls, resp. spheres, with centre $c$ and radius $r$.} Lastly, $\abs{\cdot}$ denotes the Euclidean norm, and $\norm{\cdot}$ an operator norm.
	
	\subparagraph{The closest point projection and related notions}
	The projection of points in the ambient space $\R^d$ to the {(set of)} closest point(s) of the set $\Su\subseteq \R^d$ is denoted by $\pi_{\Su}$, and illustrated in Figure~\ref{fig:pi_S}.
	
	\begin{figure}[h!]
		\centering
		\includegraphics[width=.55\textwidth]{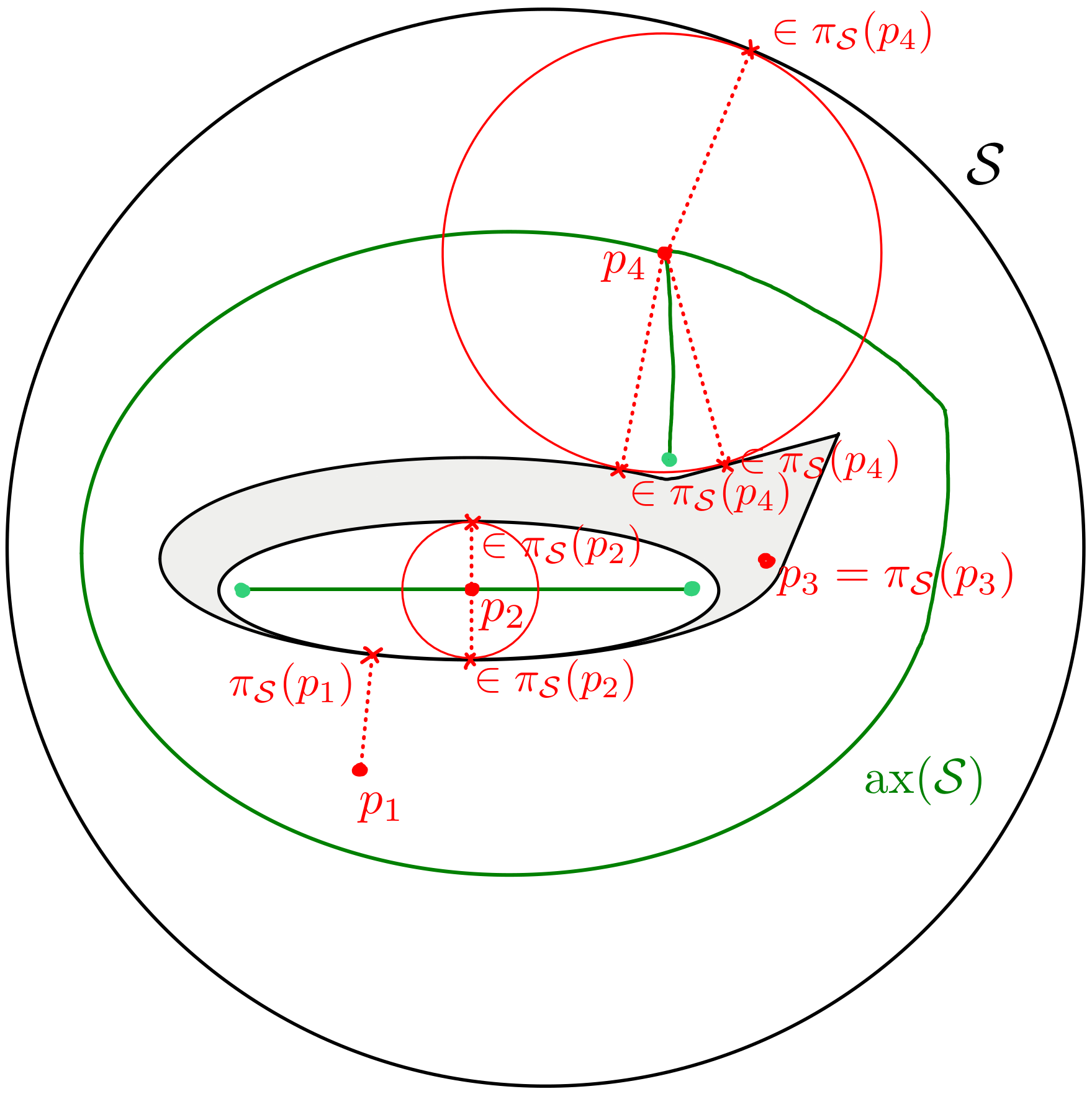}
		\caption{\small The closest point projection to the set $\Su$ of four points in $\R^2$. When a point lies on the medial axis $\ax(\Su)$, the closest point projection consists of more points.  }
		\label{fig:pi_S}
	\end{figure}
	
	The \emph{medial axis} of $\Su$ is the set of all points $p\in\R^d$ where the set $\pi_{\Su}(p)$ consists of more than one point:
	\[
	\ax(\Su) = \left\{ p\in\R^d \mid \#\pi_{\Su}(p)>1 \right\}.
	\]
	Here, $\#\pi_{\Su}(p)$ denotes the cardinality of the set $\pi_{\Su}(p)$.

	For a point $p\in \Su$, the \emph{local feature size of $p$} is the distance from $p$ to the medial axis of the set $\Su$:
	\[
	\lfs(p) = d(p, \ax(\Su)).
	\]
	Finally, the \emph{reach} of the set $\Su$ is the infimum of the local feature size over all its points:
	\[
	\rch(\Su) = \inf_{p\in \Su} \lfs(p) = \inf_{p\in \Su}d(p, \ax(\Su)).
	\]

	\textbf{\emph{Throughout this paper we assume that $\Su\subseteq \R^d$ is a closed set.}}
	We shall further assume that the set $\Su$ as well as its medial axis are bounded, and that the bounding sphere of $\Su$ is contained in $\Su$ itself. 
	More specifically, we assume that there exists a closed ball $B$ of positive radius such that $\Su \subseteq B$, and  $\partial B \subseteq \Su$. We call $\partial B$ the bounding sphere of $\Su$. 

	The addition of the bounding sphere $\partial B$ to the set $\Su$ is necessary to obtain the desired bound on the Hausdorff distance between the two medial axes of the set $\Su$ and its image under the ambient diffeomorphism.
	Indeed, consider the following example, illustrated in Figure~\ref{fig:infinite_Haus_distance}.

	Let the set $\Su$ consist of two points in the plane, $\Su= \{p ,q\} \subseteq \mathbb{R}^2$.
	The medial axis of $\Su$ is then the bisector line of $p$ and $q$. After a generic perturbation $F$ of $p$ and $q$ --- that is, not a translation and not a perturbation in the direction $\pm (p-q)$ --- the bisector line $\ax(F(\Su))$ of the perturbed points intersects the bisector $\ax(\Su)$ of the original pair. The Hausdorff distance between these two non-parallel lines is infinite, and thus unboundable. {If, however, we restrict ourselves to a ball around the origin of size $r$, the Hausdorff distance between the two restricted medial axes is of order $\mathcal{O} (r \theta)$, where $\theta$ is the angle between the vectors $p-q$ and $F(p)- F(q)$.}
	
	\begin{figure}[h!]
		\centering
		\includegraphics[width=.65\textwidth]{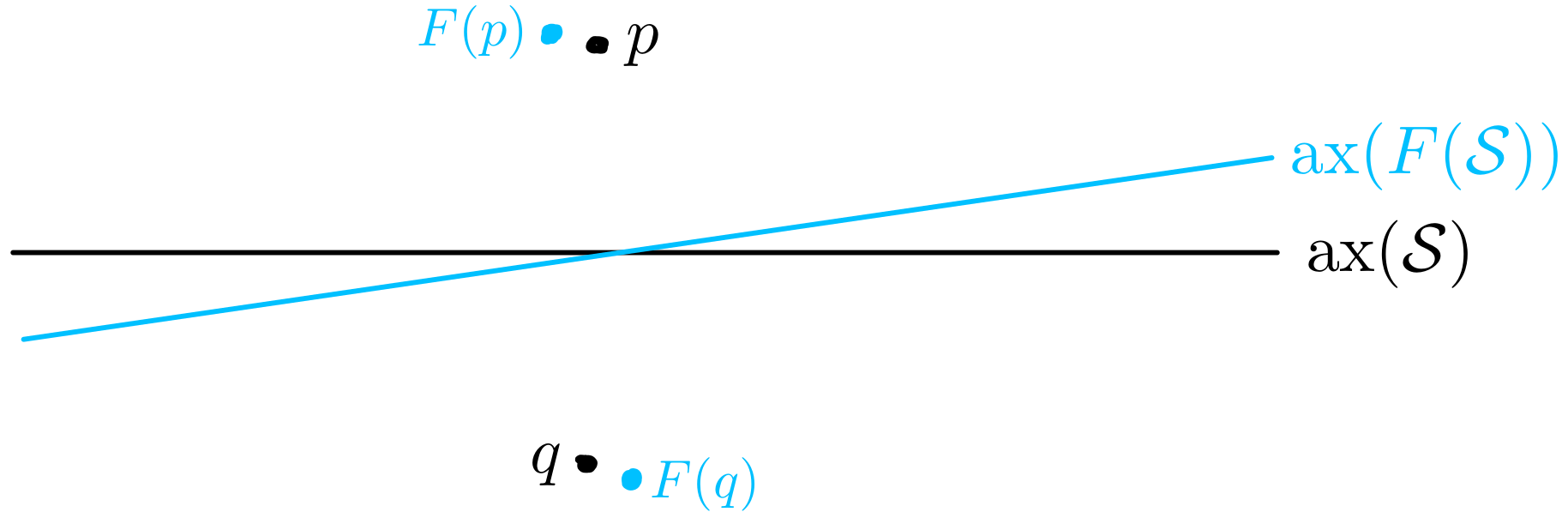}
		\caption{\small In black the set $\Su$ and its medial axis, in light blue the perturbed set and its medial axis. {Since the lines $\ax(\Su)$ and $\ax(F(\Su))$ are non-parallel, the Hausdorff distance between them is infinite. Hence it is impossible to give a bound on the distance between the medial axes without localizing. 
			} 
		}
		\label{fig:infinite_Haus_distance}
	\end{figure}
	
	{
		At the same time, the addition of the bounding sphere $\partial B$ to the considered set $\Su$ is not a restriction. Indeed, 
		\begin{remark} \label{rem:BoundingSphereNoProblem}
			The medial axes of $\Su$ and $\Su \setminus \partial B$ coincide in the interior of the ball $B$ sufficiently far away from its boundary $\partial B$. 
			More precisely:
			\begin{itemize}
				\item Any point $x \in \ax (\Su)$, such that $\pi_\Su(x) \cap \partial B = \emptyset$, lies on the medial axis $\ax(\Su \setminus \partial B)$.
				\item Conversely, if a point $x$ lies on the medial axis $\ax(\Su \setminus \partial B)$, and any (and thus every) point $q\in \pi_{\Su \setminus \partial B} (x) $ satisfies $d(x,q) < d(x, \partial B) $, then $x\in  \ax(\Su)$. 
			\end{itemize}
			Thus, the medial axis is locally stable if the ambient diffeomorphism is close to the identity.\footnote{The bounding sphere does allow one to give a relatively clean mathematical statement, see Section \ref{sec:BanachHolder}.
			} 
		\end{remark} 
	}
	
	A recurring strategy in this article is to start at a point $p$ on the set $\Su$, move away from this point in a `normal' direction, {and see if by projecting using the closest point projection $\pi_\Su$ we get back to $p$}. 
	To this end, we define the \emph{projection range}.
	
	\begin{definition}[Projection range]\label{def:projection_range}
		Let $p \in \Su$ be a point and $v \in \mathbb{R}^d$ a vector. The \emph{projection range} $d(p,v,\pi_\Su)$ in direction $v$ is the maximal distance one can travel from $p$ along $v$ such that the closest point projection yields only the point $p$:
		\[
		d(p,v,\pi_\Su) = \sup \{ \lambda \in \mathbb{R} \mid \pi_\Su (p + \lambda v )= \{p\} \}.
		\]
	\end{definition}
	
	Since $\pi_\Su(p) = \{p\}$, the projection range is canonically non-negative. {Furthermore, the directions for which the range is positive are key to our study, because of the following property:}
	\begin{lemma}[Theorem 4.8 (6) of \cite{Federer}]\label{Fed4.8.6}
		Consider a point $p \in \Su$ and a vector $v \in \mathbb{R}^d$. If
		\begin{align}
			0 <  d(p,v,\pi_\Su) < \infty,
			\nonumber
		\end{align} 
		then $p + d(p,v,\pi_\Su)\cdot v \in  \overline{\ax (\Su)}$, where 
		$\overline{X}$ denotes the closure of $X$.
	\end{lemma}
	
	{ 
		We call these special directions $v$ \emph{back projection vectors}:
		\begin{definition}[Unit back projection vectors]\label{def:UNor}
				
			For a point $p\in\Su$, $\UBP(p,\Su)$ is the set of unit vectors with a positive projection range:
			\[
			\UBP(p,\Su) = \left\{ u\in\R^d\mid \abs{u} = 1 \text{ and } 0< d(p,u,  \pi_\Su  )<\infty \right\}.
			\]
			We further define
			\begin{align*}
				\UBP(\Su) &= \left\{ (p,u)\in\Su\times\R^d \, \middle | \,   u\in \UBP(p,\Su) \right\}, 
				\\
				\BP(\Su) &= \left\{ (p,\lambda u)\in\Su\times\R^d \, \middle | \, (p,u)\in \UBP(\Su), \lambda\geq 0 \right\}.
			\end{align*}
		\end{definition}
		
		\begin{figure}[h!]
			\centering
			\includegraphics[width=.45\textwidth]{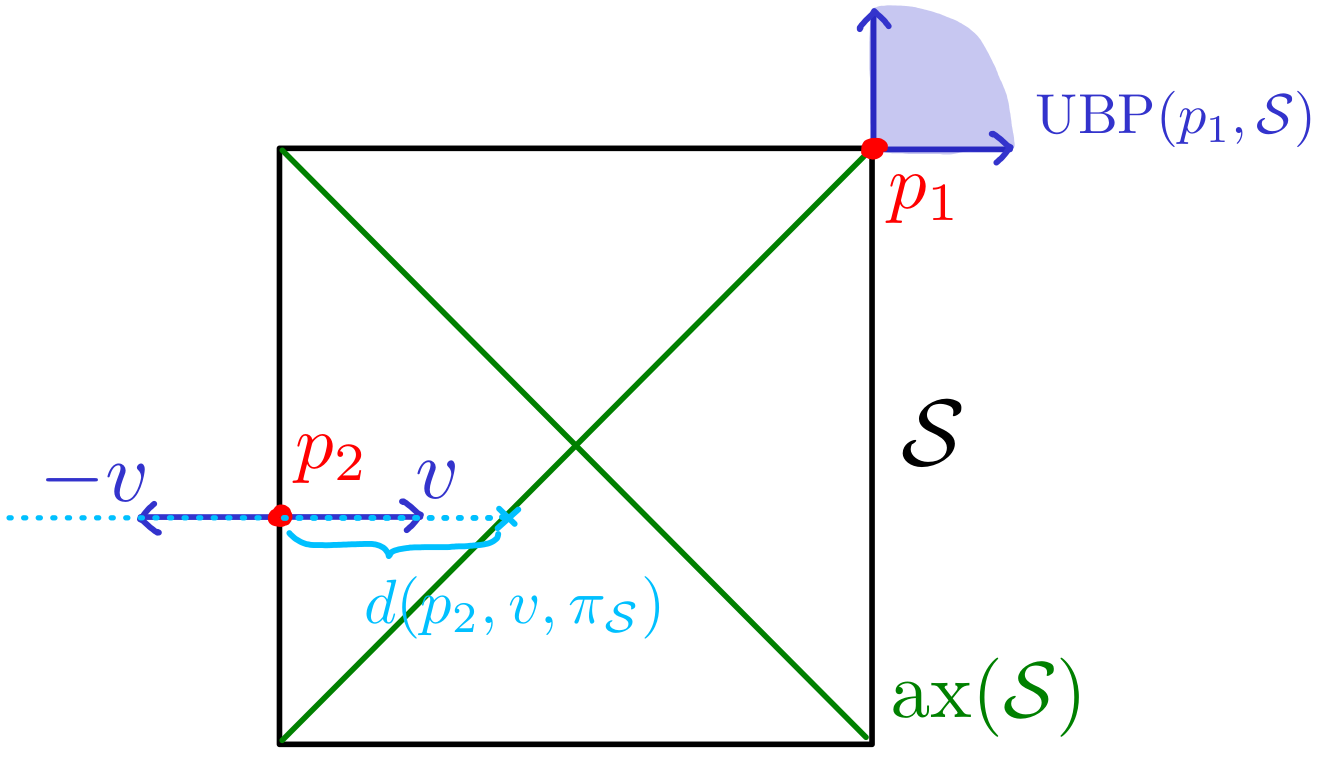}
			\caption{\small {The projection range $d(p_1,w,\pi_\Su)$ equals infinity if $w\in \UBP(p_1,\Su)$ and zero otherwise. For $p_2, d(p_2,-v,\pi_\Su)$ is infinite, $d(p_2,v,\pi_\Su)$ is finite and non-zero, and $d(p_2,w,\pi_\Su) = 0$ for $w\neq \pm v$. Hence, the set of unit projection vectors equals $\UBP(p_2,\Su) = \{v, -v\}$. }
			}
			\label{fig:square}
		\end{figure}
		%
		Thanks to Lemma~\ref{Fed4.8.6},  the following map is well-defined:  
		\begin{equation} 
			\pi_{\ax,\Su}: \UBP(\Su)  \to \overline{\ax ( \Su)},  
			\qquad (p,u) \mapsto p + d(p,u,  \pi_\Su  ) u.
			\label{eq:def_axis_projection}
		\end{equation}
	}
	
	\subparagraph{The generalized tangent and normal space}
	Back projection vectors are intricately related to the generalized tangent and normal spaces. 
	\begin{definition}[Definitions 4.3 and 4.4 of \cite{Federer}] \label{def:4.3and4.4Fed} 
		Let $p\in \Su$. The \emph{generalized tangent space} $\Tan(p,\Su)$ is the set of vectors $u \in \R^d$, such that either $u=0$ or, for every $\varepsilon>0$ there exists a point $q \in  \Su$ with
		\begin{align}
			0<&|q-p|<\varepsilon &\textrm{and} & &\left| \frac{q-p}{|q-p|}- \frac{u}{|u|} \right| < \varepsilon.
			\nonumber 
		\end{align}
		{The \emph{generalized normal space}
			$
			\Nor(p,\Su)
			$
			consists of vectors $v \in \R^d$ such that  $ \langle v, u \rangle \leq 0$ for all $u \in \Tan(p,\Su)$. Vectors contained in the generalized tangent, resp. normal, space are called \emph{tangent}, resp. \emph{normal}, \emph{to $\Su$ at $p$.}}
		
	\end{definition}
	The generalized tangent and normal spaces are illustrated in Figure~\ref{fig:tangent_normal_space}.
	
	\begin{figure}[h!]
		\centering
		\includegraphics[width=.65\textwidth]{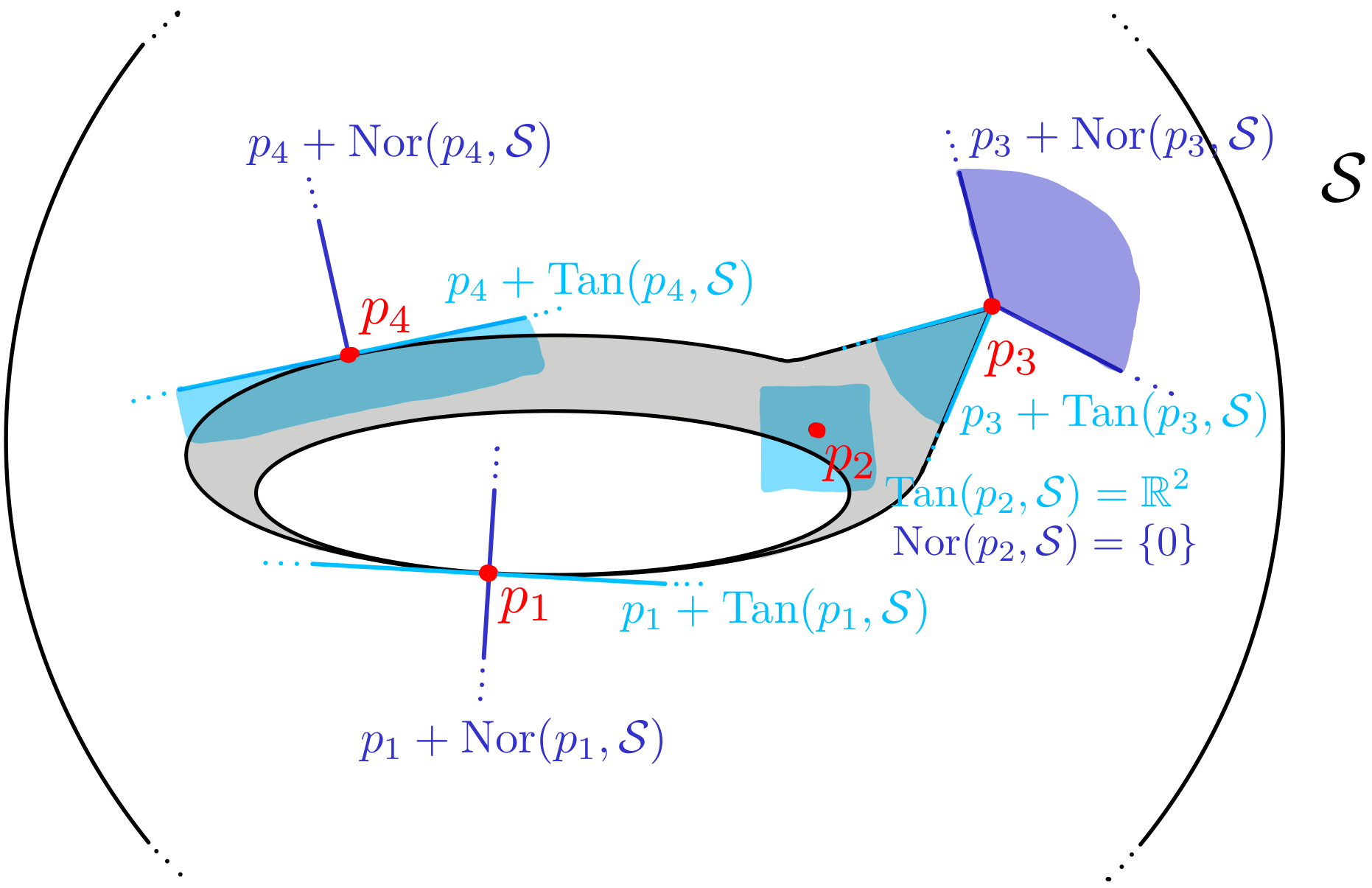}
		\caption{\small The (affine) generalized tangent and normal spaces of four points in the set $\Su\subset \R^2$, in light blue and violet, respectively. }
		\label{fig:tangent_normal_space}
	\end{figure}

	\subparagraph{Stability of the reach under ambient diffeomorphisms}
	Our last ingredient is the following result by Federer. We denote derivative by $D$. 
	\begin{theorem}[Stability of the reach under ambient diffeomorphisms, Theorem 4.19 of \cite{Federer}] \label{Fed4.19}
		{Pick two constants $0<t<\rch(\Su)$ and $s>0$. If the map
			\[ 
			F: \{ x \in \mathbb{R}^{d} \mid d(x, \Su) < s\} \to \mathbb{R}^n
			\]
			is injective and continuously differentiable, and the maps $F$, $F^{-1}$, and $DF$ are Lipschitz continuous with Lipschitz constants  $\Lip(F)$, $\Lip(F^{-1})$, $\Lip(DF)$, respectively, then the reach $\rch( F (\Su))$ of the image of the set $\Su$ under the map $F$ is lower-bounded~by 
			\[
			\rch( F (\Su)) \geq \min \left\{ \frac{s}{ \Lip(F^{-1}) } , \frac{1}{ \left(\frac{\Lip(F)} {t} + \Lip(DF) \right) \left(\Lip(F^{-1})\right)^2} \right \}.
			\]}
	\end{theorem}

	\section{Stability of the medial axis under ambient diffeomorphisms}\label{AmbientDiff}
	In this section we 
		present 
	the main result of this paper, Theorem \ref{ambientStability}. This theorem extends earlier work by Chazal and Soufflet \cite{Chazal2004}. 
	Its proof relies on Federer's result \cite{Federer} on the stability of the reach, Theorem \ref{Fed4.19}.
	It was a surprise to the authors that no assumption on the set (apart from closedness) needed to be made, and that the techniques used were that {simple and well established; they go back to Federer. 
	In fact, the authors at first envisioned a far more elaborate argument assuming the set had positive $\mu$-reach \cite{chazal2009sampling}.  
	
	To give a more geometrical interpretation to some of the results by Federer we introduce the 
	concept of a weakly tangent sphere and ball, and a maximal empty weakly tangent ball.
	\begin{definition}[Weakly tangent sphere and ball]\label{def:weakly_tangent}
		Let $p \in \Su $. A sphere 
		is called \emph{weakly tangent to $\Su$ at $p$} if it contains the point $p$ and its centre lies in the (translated) generalized normal space $\Nor(p,\Su) +p$. In other words, spheres weakly tangent to $\Su$ at $p$ are spheres with centres $p+v$ and radii $\abs{v}$, for a vector $v\in \Nor(p,\Su)$.
		
		A ball is called \emph{weakly tangent to $\Su$ at $p$} if its boundary sphere is \emph{weakly tangent to $\Su$ at $p$}.
	\end{definition}
	
	\begin{remark}\label{rem:AlternativeDef} 
		Using the definition of $\Nor(p,\Su)$, a weakly tangent ball can also be defined as follows: A ball $B(c,r)$ is weakly tangent at $p$ if and only if its centre $c$ and radius $r$ satisfy
		\[ 
		\left(p+\Tan (\Su,p)\right) \cap B(c,r)= \{p\}.
		\]
	\end{remark}
	
	The following lemma (\ref{Lem:PointsBeyondTheMedialAxisNoLongerProjectOnTheSamePoint}) essentially tells us that a family of weakly tangent balls $\left\{ B(p + \lambda v, \lambda|v|) \right\}_{\lambda\geq 0}$ contains at most one
	which is maximal with respect to inclusion among those whose interior is disjoint from the set $\Su$. Two such families are illustrated in Figure~\ref{fig:nested_balls}.
	\begin{lemma}\label{Lem:PointsBeyondTheMedialAxisNoLongerProjectOnTheSamePoint}
		Let $p\in \Su$ and $v \in \mathbb{R}^d$, and suppose that for some $\lambda>0$ we have $\pi_\Su( p + \lambda v) \neq \{p\}$. Then, for all $\lambda' \geq \lambda$, we have $\pi_\Su( p + \lambda' v) \neq \{p\}$
		{and for all $\lambda' > \lambda$, that $ p \notin \pi_\Su( p + \lambda' v)$.  }
	\end{lemma}
	\begin{proof}
		We first note that the statement is empty if $v=0$. Let $v\neq 0$, and consider the nested family of balls $B(p + \lambda' v, \lambda'|v|)$, parametrized by $\lambda'>0$.

		Because $\pi_\Su( p + \lambda v) \neq p$, the (closed) ball $B(p + \lambda v, \lambda |v|)$ contains a point $q\in\Su$ other than $p$. Since the balls $B(p + \lambda 'v, \lambda'|v|)$ are nested, the point $q$ lies inside every ball $B(p + \lambda' v, \lambda'|v|)$ with $\lambda' \geq\lambda$. Moreover, $q$ lies in the interior of $B(p + \lambda' v, \lambda'|v|)$ for $ \lambda' > \lambda$. 
		Hence, for every $\lambda' \geq\lambda$, we have that $\pi_\Su( p + \lambda' v) \neq \{p\}$ and for $\lambda' > \lambda$, that $ p \notin \pi_\Su( p + \lambda' v)$. 
	\end{proof}
	
	\begin{figure}[h!]
		\centering
		\includegraphics[width=.6\textwidth]{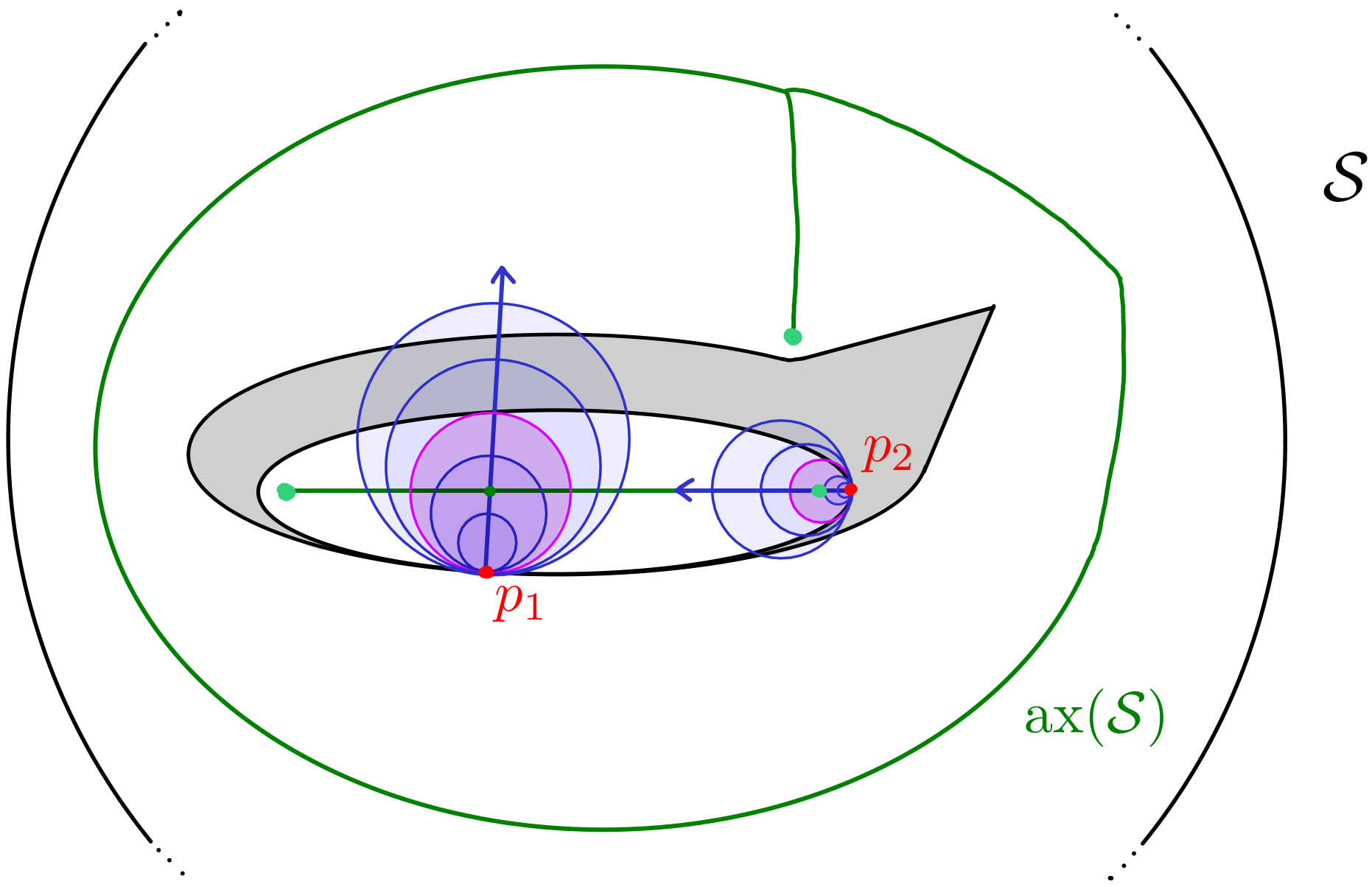}
		\caption{\small Two families of balls weakly tangent to the set $\Su\subset \R^2$ (in blue). Each family contains a unique maximal empty ball (in purple). Notice that the centre of the maximal empty ball weakly tangent at the point $p_1$ lies at the medial axis $\ax(\Su)$, while the centre of the maximal empty ball weakly tangent at the point $p_2$ only lies at its closure, $\overline{\ax ( \Su)}$. }
		\label{fig:nested_balls}
	\end{figure}

	The maximal ball in the family of weakly tangent balls $\left\{ B(p + \lambda v, \lambda|v|) \right\}_{\lambda\geq 0}$ with respect to inclusion (as described in 
	Lemma~\ref{Lem:PointsBeyondTheMedialAxisNoLongerProjectOnTheSamePoint})  
	is called \emph{maximal empty}. 
	{For the purpose of this article, we define maximal empty balls in terms of unit back projection vectors (Definition~\ref{def:UNor}). To see that each maximal empty ball is indeed weakly tangent, we emphasise:
		
		\begin{lemma}\label{lem:back_projection_subset_normal_cone}
			If $(p,v) \in  \BP(\Su)$, then $(p,v) \in \Nor(\Su)$. That is, 
			$
			\BP(\Su) \subseteq  \Nor(\Su)
			$. 
			In particular, for any pair $(p,u) \in { \UBP(\Su)}$ and radius $\lambda\geq 0$, the ball $B(p+\lambda u, \lambda)$ is weakly tangent to~$\Su$.
		\end{lemma}
		
		
	} 
	
	\begin{remark}
		For general closed sets, the converse of Lemma~\ref{lem:back_projection_subset_normal_cone}, that is, $ \Nor(\Su) \subseteq \BP(\Su)$, is not true. One counter-example is the graph of the function $x\mapsto |x|^{3/2}$ at the origin. However, the inclusion $ \Nor(\Su) \subseteq \BP(\Su)$ holds for sets of positive reach thanks to \cite[Theorem 4.8 (12)]{Federer}, 
		which we recall as Lemma~\ref{Fed4.8.12} in the appendix. 
	\end{remark}

	\begin{definition}[Maximal empty weakly tangent ball]\label{def:MaximalEmptyTangentBall}
		Let $(p,u) \in { \UBP(\Su)}$. A weakly tangent ball $B(p+\lambda u,\lambda)$ is called \emph{maximal empty} {to $\Su$} if $\lambda = d(p,u,\pi_\Su)$, or, equivalently, if $\pi_{\ax, \Su} (p,u) =p+\lambda u$.
	\end{definition}
	
		{(Maximal empty) weakly tangent balls satisfy the following properties. Let $(p,u) \in  \UBP(\Su)$.
			\begin{itemize}
				\item For any radius $0<\lambda \leq d(p,u,\pi_\Su)$, the interior of the ball $B(p+\lambda u,\lambda)$ is disjoint from the set $\Su$.
				This follows directly from Definition~\ref{def:MaximalEmptyTangentBall} and Lemma~\ref{Lem:PointsBeyondTheMedialAxisNoLongerProjectOnTheSamePoint}.  
				
				\item The centres of maximal empty weakly tangent balls lie on the closure of the medial axis of $\Su$. This is due to Lemma \ref{Fed4.8.6} and the definition of the map $\pi_{\ax, \Su}$ (equation \eqref{eq:def_axis_projection}). 
			\end{itemize}
	}
	
		
		The following lemma moreover tells us, that each point on the medial axis is a centre of a maximal empty weakly tangent ball.
		\begin{lemma}[Surjectivity on $\ax (\Su)$]\label{lem:surjectivityMap}
			For any point $x \in \ax (\Su)$ and $p \in \pi_{\Su} (x)$,
			there exists a vector $u \in \UBP(p, \Su)$ such that $\pi_{\ax,\Su} (p,u)=x$. In other words, $B(x,|x-p|)$ is a maximally empty weakly tangent ball. Moreover, we have that 
			\[
			\ax (\Su)\subseteq \pi_{\ax,\Su}\left(\UBP(\Su)\right) \subseteq \overline{\ax (\Su)} .
			\]
		\end{lemma} 
		\begin{proof}
			Let $Q= \pi_{\Su} (x)$ be the subset of $\Su$ that is closest to $x$. 
			Because $x \in \ax(\Su)$, $Q$ contains at least two points{, one of them being $p$}.
			We write $\lambda = \abs{x-p}$. Since $\Su$ and $\ax(\Su)$ are disjoint, $\lambda>0$, and thus we can define $u = \frac{x-p}{\lambda}$.
			
			Since the interior of the ball $B(x, \lambda)$ does not intersect $\Su$, it in particular does not intersect $\Tan(p,\Su)$ and thus $B(x, \lambda)$ is weakly tangent at $p$ by Remark \ref{rem:AlternativeDef}.
			Let us now consider the nested family 
			$
			B (p + \lambda'u , \lambda')
			$
			of weakly tangent balls at $p$. 
			By definition, $\partial B(x,\lambda) \cap \Su= Q$ and therefore $ B (p + \lambda'u , \lambda') \cap \Su= p$ for $\lambda'< \lambda $. At the same time, Lemma \ref{Lem:PointsBeyondTheMedialAxisNoLongerProjectOnTheSamePoint}
			yields that for $\lambda'> \lambda $, {$p \notin \pi_\Su (p + \lambda'u ) $.} 
			Hence the projection range in direction $u$ equals $ d(p,u,\pi_\Su) = \lambda$ and we obtain $\pi_{\ax,\Su} (p,u)=p+\lambda u =x$ directly from Definition \ref{def:MaximalEmptyTangentBall}
			.  The fact that $\pi_{\ax,\Su}\left(\UBP(\Su)\right) \subseteq \overline{\ax (\Su)}$ is due to Lemma \ref{Fed4.8.6}.
		\end{proof}

		We are now {almost} ready to state our main theorem. Before phrasing the result, we walk the reader through the assumptions and fix the notation on the way. The assumptions are illustrated in Figure~\ref{fig:settings_theorem}.
		\begin{figure}[h!]
			\centering
			\includegraphics[width=\textwidth]{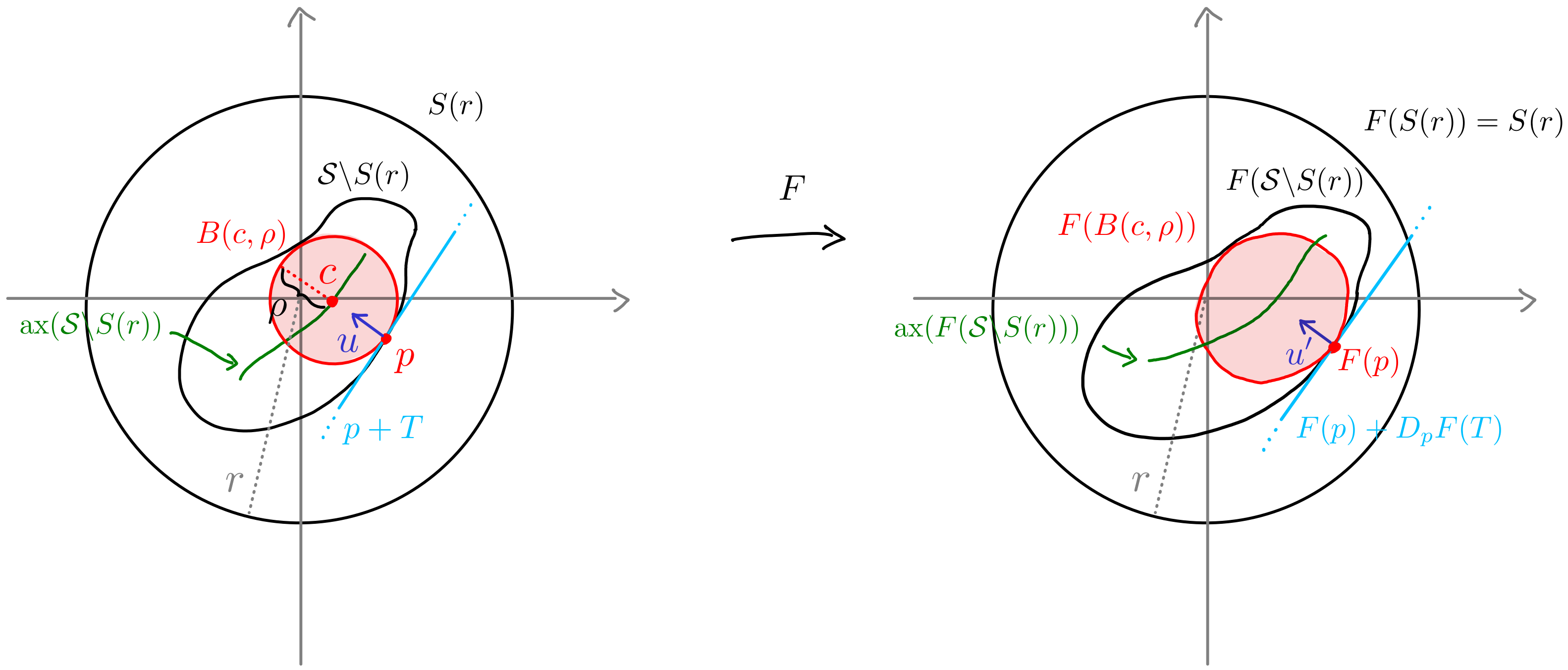}
			\caption{\small The setting of Theorem~\ref{ambientStability}.}
			\label{fig:settings_theorem}
		\end{figure}
		
		\begin{assumption}
			$\phantom{3}$ 
			\begin{itemize}
				\item We assume that the set $\Su$ has a bounding sphere of radius $r$, which we denote by $S(r)$.
				\item We consider a $C^{1}$~diffeomorphism $F:\R^d\to \R^d$ such that the Lipschitz constants of $F$ and $F^{-1}$ are bounded by $L_F$, and the Lipschitz constants of the differentials $DF$ and $DF^{-1}$ are bounded by $L_{DF}$. We call such a diffeomorphism a $C^{1,1}$ diffeomorphism. 
				\item We further assume that the map $F$ leaves the bounding sphere $S(r)$ invariant, that is, $F(S(r)) = S(r)$.
				\item { 
					We pick a point $c \in \ax (\Su)$, 
					a point $p \in \pi_{\Su} (c)$, and write $\rho = |c-p|$. Observe that since $\Su \cap \ax (\Su) = \emptyset$, $\rho$ is positive.
					By Lemma \ref{lem:surjectivityMap}, the ball $B(c,\rho)$ is a maximal empty weakly tangent ball to $\Su$ at $p$. Moreover, we define  $u= \frac{c-p}{\abs{c-p}}$ and note that $u \in {\UBP} (p,\Su)$. 
				} 
				
				\item We denote the tangent hyperplane to the boundary sphere of $B(c,\rho)$ at $p$ by $p+T$. The hyperplane $T$ is the orthocomplement of the vector $u$: $T = u^\perp$. 
				\item 
					We work with the unit vector at $F(p)$ that points inside the image of the ball $B(c,\rho)$ and is orthogonal to the hyperplane $D_p F(T)$. We denote this vector by $u'$.
				\end{itemize}
			\end{assumption}
			
			{ 
				\begin{remark}
					Thanks to Remark \ref{rem:BoundingSphereNoProblem}, the first and the third assumption can be replaced by the following: 
					\begin{itemize} 
						\item An assumption on the density of $\Su$: Every point in $\mathbb{R}^d$ is at most a distance $r$ from some point in $\Su$, and this property is preserved by $F$. 
						{In other words, all $x \in \R^d$ satisfy $|x -\pi_\Su (x)| \leq r$ as well as $|x-\pi _{F(\Su)}(x)| \leq r$.}
						\item We do not have to explicitly refer to the bounding sphere. That is, our bounds hold for any ball $B(c, r/2)$ satisfying $B(c, r/2-L_F r) \cap \Su \neq \emptyset$.
					\end{itemize}   
				\end{remark} 
			}

			The following properties of $u$ and $u'$ play an important rule in the proof of the theorem:
			
			\begin{claim}\label{claim:adjoint_operator}
				\begin{align} 
					u' = \frac{(D_p F^t)^{-1} (u)}{| (D_p F^t)^{-1} (u)|},
					\label{eq:Def_u_prime} 
				\end{align}     
				where $D_pF^t$  is the transpose matrix (or the adjoint operator) of $DF$ at the point $p$, defined by 
				\begin{align} 
					\forall v_1,v_2,\qquad \langle v_1,  D_p F (v_2) \rangle =  \langle D_pF^t (v_1), v_2 \rangle. \nonumber 
				\end{align}  
			\end{claim} 
			
			\begin{claim}\label{lemma:BoundAngleUUPrime}
				Let $\|D_pF - \operatorname{Id}\| \leq \varepsilon <1$. Then the angle $\angle u, u'$ between the vectors $u$ and $u'$ satisfies 
				\[
				\cos \angle u, u' \geq \sqrt{1 - \varepsilon^2}.
				\]
			\end{claim}
			
			The proof of both of these properties is an exercise in linear algebra, which we defer to the appendix.
			
			
			\begin{theorem} \label{ambientStability} 
				Under the above assumptions, there exists a maximal empty weakly tangent ball  $B({ c'} ,\rho')$ 
				{ to the set $F(\Su)$,} 
				whose boundary sphere has an internal normal $u'$. 
				In particular, the ball $B({ c'} ,\rho')$ is tangent to the affine hyperplane $F(p)+D_p F (T)$. Its radius $\rho'$ is bounded by
				$
				\rho'  \in \left[\frac{\rho}{ (L_{F})^3  + \rho L_{DF} (L_{F})^{2}},\frac{ (L_{F})^3  \rho}{ 1 - \rho L_{DF} (L_{F})^{2}} \right]. 
				$
				Assume, moreover, that the distortions of both $F$ and $DF$ are bounded, that is, for all $x\in\R^d$,
				\begin{align} 
					|F(x)-x| 
					\leq \varepsilon_1 ,\qquad
					\label{CloseToId}
					\|DF_x - \operatorname{Id}\| 
					\leq \varepsilon_2 <1,
				\end{align}
				and {$r \cdot L_{DF} (L_F)^2 \leq 1/2$}. Define 
				\begin{align} \label{eq:C_L}
					&C_L(r,L_{F}, L_{DF},\varepsilon_1,\varepsilon_2 ) = 
					\nonumber 
					\\ 
					&2 r \sqrt{1 +  (L_{F})^6  \left ( 1+4 r L_{DF} (L_{F})^{2} \right)^2   - 2    (L_{F})^3  \left( 1+4r L_{DF} (L_{F})^{2} \right)   \sqrt{1- (\varepsilon_2)^2} }  +\varepsilon_1 
				\end{align} 
				then the map $\pi_{\ax,\Su}$ satisfies 
				\[|\pi_{\ax,\Su} (p, u)- \pi_{\ax,F(\Su)} (F(p), u' 
				)| 
				\leq C_L(r,L_{F}, L_{DF},\varepsilon_1,\varepsilon_2 ). 
				\] 
				Thus, the Hausdorff distance between the medial axes of $\Su$ and its image $F(\Su)$ is bounded by
				\begin{align}
					d_H (\ax(\Su), \ax ( F(\Su ) )) \leq C_L(r,L_{F}, L_{DF},\varepsilon_1,\varepsilon_2 ).
					\label{eq:BoundHausdorffDistance}
				\end{align} 
			\end{theorem} 
			
			The bound {$|F(x)-x| \leq \varepsilon_1$}  
			really is necessary, because we want {our theorem to accommodate for} 
			rotations and translations, which rotate and translate the medial axis without changing distances and hence have Lipschitz constant 1. We further stress that if the diffeomorphism $F$ is close to the identity, its Lipschitz constant satisfies 
			$L_F\geq 1$, because by assumption $F$ leaves the bounding sphere $S(r)$ invariant, and $L_{DF}$ is close to zero. Under further assumptions, expression~\eqref{eq:C_L} can be bound explicitly, as is done in Remark~\ref{remark20}.

			{

					\begin{proof}
						We first derive the bounds for the radius $\rho'$.
						As the first step, we apply Theorem~\ref{Fed4.19}
						to the boundary sphere 
						$S(c,\rho)$ of the maximal empty weakly tangent ball $B(c,\rho)$. 
						{In particular, we can choose the constant $s$ in Theorem~2.6 
							arbitrarily large, and the constant $t$ arbitrarily close to the reach $\rch(S(c,\rho)) = \rho$, to obtain:}
						\[
						\rch\left(F(S(c,\rho))\right)\geq \frac{1}{ \left(\frac{L_{F}} {\rho} + L_{DF} \right) (L_{F})^{2}}
						= \frac{\rho}{ (L_{F})^3  + \rho L_{DF} (L_{F})^{2}}=:\rho_1 . 
						\]
						This means that no {open} ball of radius $\rho_1$ tangent to the set $F(S(c,\rho))$ actually intersects $F(S(c,\rho))$. In addition, since the set $F(B(c,\rho))$ does not contain any points of $F(\Su)$ in its interior, no ball of radius $\rho_1$ that is tangent to $F(S(c,\rho))$ and whose centre lies inside $F(S(c,\rho))$ contains any point of $F(\Su)$.

						
						{The unit vector $u'\in DF_{p}(T)^\perp$ (defined in \eqref{eq:Def_u_prime}) is defined such that the point $F(p)+\rho_1 u'$ lies inside the distorted ball $F(B(c,\rho))$.} Due to the above observation, the ball $B(F(p)+\rho_1 u', \rho_1)$ is weakly tangent to $F(\Su)$ at $F(p)$ and contains no points of $F(\Su)$ in its interior.

						{Let us now consider the weakly tangent ball $B(F(p) +\rho'' u',  \rho'')$, whose radius $\rho''$ satisfies
							\[ 
							\rho''>\frac{ (L_{F})^3  \rho}{ 1 - \rho L_{DF} (L_{F})^{2}} =:\rho_2.
							\]
							To shorten the notation, we set
							\[
							F(p) +\rho'' u'=:c''.
							\]
							
							To derive a contradiction, we assume that $B(c'',  \rho'')$ is maximal empty.
						}
						This is equivalent to assuming that $\textrm{int} B(c'', \rho'' ) \cap { F} (\Su)= \emptyset$, and thus $B(c'', \rho'')$ is a maximal empty weakly tangent ball to $F(p)$. Similarly to the beginning of the proof, we now apply Theorem \ref{Fed4.19}
						to the map $F^{-1}$ and the boundary sphere $S(c'', \rho'') = \partial B(c'', \rho'')$. 
						As a result, the reach of $F^{-1} (S(c'', \rho''))$ is at least 
						\begin{align}
							\rch\left(F^{-1} (S(c'', \rho''))\right) \geq&\frac{\frac{ (L_{F})^3  \rho}{ 1 - \rho L_{DF} (L_{F})^{2}} }{ (L_{F})^3  +\frac{ (L_{F})^3  \rho}{ 1 - \rho L_{DF} (L_{F})^{2}} L_{DF} (L_{F})^{2}} 
							= \rho.
							\nonumber
						\end{align}
						We conclude that there exists a ball that is tangent to the set $F^{-1} (S(c'', \rho''))$ at $F^{-1}(F(p))=p$, whose radius is larger than $\rho$, and that does not contain any points of $\Su$ in its interior. This contradicts  the fact that the ball $B(c,\rho)$ is maximal empty, and completes the proof of the first part of the statement.

						We now prove the bounds on the distortion of the map $\pi_{\ax, \Su}$.
						{Let $\rho'\in [\rho_1, \rho_2]$ be the radius of the maximal empty weakly tangent ball at $F(p)$ in the direction $u'$, and write $c' := F(p)+\rho' u'$ for its centre. }
						We stress that, as a consequence of {Lemma \ref{Fed4.8.6} (Theorem 4.8 (6) of \cite{Federer})
						}, $c' \in \overline {\ax (F(\Su)) }$, but it is not necessarily true that $c' \in {\ax (F(\Su)) }$. 
						
						{The goal is to estimate the distance between the two centres $c = \pi_{\ax, \Su}(p,u)$ and $c' = \pi_{\ax, \Su}(F(p),u')$.
							Indeed, since $c-p = \rho u$ and $c'-F(p) = \rho' u'$,
							\begin{align*}
								\abs{c-c'} &= \abs{c-p+p-F(p)+F(p)-c'} = \abs{\rho u + p-F(p) - \rho' c'}\\
								&\leq \abs{\rho u - \rho' u'} + \abs{F(p)-p}.
							\end{align*}
							Due to the assumptions of the theorem,
							$\abs{F(p)-p}\leq\varepsilon_1$. 
							Furthermore, thanks to Claim \ref{lemma:BoundAngleUUPrime},
							\begin{align*}
								\abs{\rho u - \rho' u'}^2 = \rho^2+(\rho')^2-2\rho\rho' \cos\angle u, u'\leq \rho^2+(\rho')^2-2\rho\rho' \sqrt{1- (\varepsilon_2)^2}.
							\end{align*}
						}
						
						{Recalling that $\rho'\in [\rho_1, \rho_2]$, we thus obtain}
						\begin{align}
							\abs{\rho u - \rho' u'}\leq&\max \left( \sqrt{\rho^2 + (\rho_1)^2  - 2 \rho \, \rho_1 \cos (\arcsin (\varepsilon_2)) },\sqrt{\rho^2 + (\rho_2)^2  - 2 \rho \, \rho_2 \cos (\arcsin (\varepsilon_2)) } \right)
							\nonumber
							\\
							&= \max \left( \sqrt{\rho^2 + (\rho_1)^2  - 2 \rho \, \rho_1 \sqrt{1-(\varepsilon_2)^2}  },\sqrt{\rho^2 + (\rho_2)^2  - 2 \rho \, \rho_2 \sqrt{1- (\varepsilon_2)^2} } \right).\nonumber
						\end{align}
						
						Hence,
						\begin{align} 
							|c- c'| \leq  \max \left( \sqrt{\rho^2 + (\rho_1)^2  - 2 \rho \, \rho_1 \sqrt{1-(\varepsilon_2)^2}  },\sqrt{\rho^2 + \left(\rho_2\right)^2  - 2 \rho \, \rho_2  \sqrt{1- (\varepsilon_2)^2} } \right) +\varepsilon_1.
							\label{Eq:BoundDistCenters}
						\end{align} 
						
						As the last step, we simplify the expression \eqref{Eq:BoundDistCenters} 
					(at the cost of weakening the bounds). For this, we assume that $\rho L_{DF} (L_F)^2 \leq 1/2$, 
					so that 
					\begin{align} 
						\rho_1
						&= \frac{\rho}{ (L_{F})^3  + \rho L_{DF} (L_{F})^{2}}  \geq \frac{\rho}{(L_F)^3} \left( 1- \rho\frac{L_{DF}}{L_F}  \right) ,
						\label{boundRho1a}
						\\
						\rho_2
						&=\frac{ (L_{F})^3  \rho}{ 1 - \rho L_{DF} (L_{F})^{2}} \leq \rho (L_{F})^3  \left ( 1+2 \rho L_{DF} (L_{F})^{2} \right), 
						\label{boundRho2a}
					\end{align} 
					where we used that, for $x\in [0,1/2]$, $\frac{1}{1+x} \geq 1-x$ and $\frac{1}{1-x} \leq 1+2x$.   
					We note that both $\rho_1$ and $\rho_2$ tend to $\rho$ as $L_F$ tends to $1$ and $L_{DF}$ tends to $0$. We now consider $|\rho_1- \rho|$ and $|\rho_2-\rho|$, and claim that 
					\[|\rho_1- \rho|,|\rho_2- \rho| \leq \rho (L_{F})^3  \left ( 1+2 \rho L_{DF} (L_{F})^{2} \right) -\rho.\] 
					For $|\rho_2- \rho| = \rho_2- \rho$, the claim holds thanks to \eqref{boundRho2a}. To establish this for $|\rho_1- \rho|$ requires a small calculation:  
					\begin{align}
						|\rho_1- \rho| = \rho- \rho_1 &\leq \rho -\frac{\rho}{(L_F)^3} \left( 1- \rho\frac{L_{DF}}{L_F}  \right) 
						\tag{due to \eqref{boundRho1a}}
						\\
						&\leq  \rho (L_{F})^3   \left ( 1+2 \rho L_{DF} (L_{F})^{2} \right) -\rho 
						\tag{assuming the claim holds}
						\\
						2 \rho &\leq \rho (L_{F})^3   \left ( 1+2 \rho L_{DF} (L_{F})^{2} \right) + \frac{\rho}{(L_F)^3} \left( 1- \rho\frac{L_{DF}}{L_F}  \right) 
						\tag{reformulating the previous inequality}
						\\
						2  &\leq  (L_{F})^3+ \frac{1}{(L_F)^3} +    2 \rho L_{DF} (L_{F})^{5}  - \rho\frac{L_{DF}}{(L_F)^4} ,  
						\nonumber
					\end{align} 
					where the final inequality holds because $x^3+x^{-3} \geq 2$, and $2 x^5-x^{-4} \geq 0$, for $x\geq 1$. We now consider the function 
					\begin{align} 
						f(\delta) &= \rho^2 + \rho^2(1+ \delta)^2  - 2 \rho^2 \, (1+ \delta) \sqrt{1-(\varepsilon_2)^2} 
						\nonumber
						\\
						&= \rho^2 \left( \delta^2  + 2  \left(1- \sqrt{1-(\varepsilon_2)^2} \right ) \delta + 2  \left (1- \sqrt{1-(\varepsilon_2)^2} \right) \right) .
						\nonumber
					\end{align}  
					The function $f$ is a second order polynomial in $\delta$ and because all coefficients are positive, the maximum of $f$ on the interval $[-\delta_m, \delta_m]$ is a attained at $\delta_m $, that is,
					\begin{align} 
						\sup_{\delta \in [-\delta_m, \delta_m] }  f(\delta) = f(\delta_m).
						\nonumber
					\end{align} 
					By combining these results, we see that 
					\begin{align} 
						&|c- c'| 
						\nonumber 
						\\
						&\leq   
						\sqrt{f \left (  (L_{F})^3  \left ( 1+2 \rho L_{DF} (L_{F})^{2} \right) -1 \right )} + \varepsilon_1
						\nonumber
						\\
						&= \sqrt{\rho^2 + \left( \rho (L_{F})^3  \left ( 1+2 \rho L_{DF} (L_{F})^{2} \right) \right)^2   - 2 \rho  \left( \rho (L_{F})^3  \left( 1+2 \rho L_{DF} (L_{F})^{2} \right) \right)  \sqrt{1- (\varepsilon_2)^2} } 
						\nonumber \\
						& \phantom{=} +\varepsilon_1
						\nonumber
						\\
						& = \rho \sqrt{1 +  (L_{F})^6  \left ( 1+2 \rho L_{DF} (L_{F})^{2} \right)^2   - 2    (L_{F})^3  \left( 1+2 \rho L_{DF} (L_{F})^{2} \right)   \sqrt{1- (\varepsilon_2)^2} }  +\varepsilon_1.
						\nonumber
					\end{align} 
					Because both $f(\delta)$ and the bound \eqref{boundRho2a} are monotone in $\rho$, {and $\rho$ is bounded by the radius $r$ of the bounding sphere $S(r)$, }
					we conclude that 
					\begin{align} 
						&|c- c'| 
						\nonumber 
						\\
						&\leq  2 r \sqrt{1 +  (L_{F})^6  \left ( 1+4 r L_{DF} (L_{F})^{2} \right)^2   - 2    (L_{F})^3  \left( 1+4r L_{DF} (L_{F})^{2} \right)   \sqrt{1- (\varepsilon_2)^2} }  +\varepsilon_1.
						\label{eq:Hausdorff_Bound_1}
					\end{align} 

					For every point $c$ in $\ax(\Su)$ we have found a point $c'$ in $\overline {\ax(F(\Su))}$ whose distance is bounded by \eqref{eq:Hausdorff_Bound_1}, and therefore the one-sided Hausdorff distance between the two medial axes $\ax(\Su)$ and $\overline {\ax(F(\Su))}$ is bounded by the same quantity. 
					Because the symmetrical formulation of the statement, the same bound holds for the Hausdorff distance. 
				\end{proof}

			}

		\section{Quantifying  \texorpdfstring{$C^{1,1}$}{differentiable with Lipschitz derivative} diffeomorphisms as deviations from identity} \label{sec:BanachHolder}
		
		In this section we reformulate the main result in terms of norms on Banach spaces. This reformulation offers a more theoretical insight, and we believe the reformulated bounds are easier to work with in certain applications. 
		{
			Indeed, in the context of practical numerical computations, a bound on the Lipschitz constant of an operator --- or, at least, a modulus of continuity --- allows to control the condition number.
			This control is particularly useful when we calculate with objects such as the medial axis, 
			whose (numerical) stability is often
			problematic in practice.
			
		}
		As we will see below, for this reformulation we somewhat strengthen our assumptions.
		
		We decompose a diffeomorphism $F$ into the identity map $\mathds{1}_{\R^d}$ on $\R^d$, and a displacement field $\varphi$: $ 
		F= \mathds{1}_{\R^d} + \varphi$.
		For the choice of the displacement field, we restrict ourselves to the vector space $\mathcal{U}$ of all $C^{1,1}$ maps $\varphi$ from $\R^d$ to $\R^d$ whose restriction to the exterior $\R^d \setminus B(r)$ of a certain bounding ball $B(r)$ equals $0$.\footnote{This is more restrictive than assuming that the restriction to the bounding sphere $S(r)$ is $0$, but it simplifies matters in this section.}
		
		A natural norm associated to $\mathcal{U}$ is one that makes it a Banach space. 
		A typical choice, inherited from general Banach spaces of $C^{1,1}$ functions,
		would be for example, for $\varphi \in \mathcal{U}$,
		\begin{equation}\label{eq:normLip11Standard}
			\| \varphi \|_{C^{1,1}} = \max \left( \| \varphi\|_\infty , \,  \| D\varphi \|_\infty, \, \operatorname{Lip}(D\varphi) \, \right).
		\end{equation}
		Here we used the following notation: 
		\begin{itemize} 
			\item $\| \varphi\|_\infty= \sup_{x\in \R^d} | \varphi(x) |$ denotes the sup norm on $x \mapsto  | \varphi(x) |$, 
			where $| \cdot | $ is the Euclidean norm in $\R^d$,
			\item  $\| D \varphi\|_\infty= \sup_{x\in \R^d} \| D\varphi(x) \|$
			denotes the sup norm on  $x \mapsto \|D \varphi(x) \|$ , where $ \|D \varphi(x) \|$ is the operator norm induced by the Euclidean norm on $\R^d$. 
			\item We write $ \operatorname{Lip}(D\varphi)$ for the Lipschitz semi-norm of $D\varphi$. 
			The Lipschitz semi-norms of $\varphi$ and $D\varphi$ are defined as
			\[
			\operatorname{Lip}(\varphi)= \sup_{x,y \in \R^d,\, x \neq y} \frac{|\varphi(y) - \varphi(x)|}{|y-x|},
			\]
			and
			\[
			\operatorname{Lip}(D\varphi)= \sup_{x,y \in \R^d,\,  x \neq y} \frac{\|D\varphi(y) - D\varphi(x)\|}{|y-x|}. 
			\]
		\end{itemize} 
		The norm defined in \eqref{eq:normLip11Standard} makes $\mathcal{U}$ into a Banach space, since every Cauchy sequence in $\mathcal{U}$ has a limit in $\mathcal{U}$. 
		In addition, any function $\varphi\in \mathcal{U}$ satisfies:
		\begin{align}
			\Lip(\varphi) &=  \| D\varphi \|_\infty \label{eq:NormComparizon_1}, \\
			\| D\varphi \|_\infty &\leq r \operatorname{Lip}(D\varphi), \label{eq:NormComparizon_2} \\
			\| \varphi\|_\infty &\leq r \operatorname{Lip}(\varphi) \leq r^2 \operatorname{Lip}(D\varphi), \label{eq:NormComparizon_3}
		\end{align}
		since the  restriction of $\varphi$ to $\R^d \setminus B(r)$ is $0$. 
		This in turn yields that 
		$\operatorname{Lip}(D\varphi)  \leq \| \varphi \|_{C^{1,1}} \leq \max(1, r, r^2)\operatorname{Lip}(D\varphi). $
		Thus, in $\mathcal{U}$, the norm $\varphi \mapsto \operatorname{Lip}(D\varphi)$
		is equivalent to the norm $\varphi\mapsto \| \varphi \|_{C^{1,1}}$.
		
		We can now state 
			slightly less general version of Theorem \ref{ambientStability} in terms of the Banach space $(\mathcal{U},\:  \varphi \mapsto \operatorname{Lip}(D\varphi)\, )$.

			\begin{theorem} \label{ambientStabilityWithLipNorm} 
				Let $\Su\subseteq\R^d$ be bounded by the ball $B(r)$ of radius $r>0$, such that $S(r) = \partial B(r)\subseteq \Su$. Let further $F$ be a $C^{1,1}$ diffeomorphism from $\mathbb{R}^d$ to itself that leaves the set $\R^d \setminus B(r)$ invariant, and define 
				two displacement fields $\varphi, \tilde{\varphi} \in \mathcal{U}$ such that 
				$F= \mathds{1}_{\R^d} + \varphi$ 
				and \[( \mathds{1}_{\R^d} + \tilde{\varphi})\circ ( \mathds{1}_{\R^d} + \varphi) =   \mathds{1}_{\R^d}.\]
				
				Define $\varepsilon = \max \left( \operatorname{Lip}(D\varphi), \operatorname{Lip}(D\tilde{\varphi})   \right)$.
				
				If $r \varepsilon  \leq 1/4$, the Hausdorff distance between the medial axes of the set $\Su$ and its image $F(\Su)$ is bounded by
				$ 
				d_H (\ax(\Su), \ax ( F(\Su ) )) \leq  \left( 1 + \sqrt{50}\right) r^2 \varepsilon +\bigO\left(r^3\varepsilon^2 \right)
				$. 
				In particular,
				$ 
					d_H (\ax(\Su), \ax ( F(\Su ) )) 
					{= \bigO\left( r^2 \varepsilon \right) }
					$. 
				\end{theorem}
				
				\begin{proof}
					We denote $L_{\varphi} = \Lip(\phi)$. 
					Expressions 
					\eqref{eq:NormComparizon_1}, \eqref{eq:NormComparizon_2} and \eqref{eq:NormComparizon_3} 
					yield:
					\begin{align}\label{eq:RelationEpsilonLipFDF}
						L_{\varphi} &\leq  r \varepsilon ,\qquad
						L_{DF} = \varepsilon ,\qquad
						L_{F} \leq 1 + L_{\varphi} \leq 1  +  r \varepsilon,\qquad
						\varepsilon_1 \leq r^2 \varepsilon ,\qquad
						\varepsilon_2 \leq r \varepsilon.
					\end{align}

					We deduce
					\[
					r \varepsilon  \leq {1/4} \Longrightarrow r \varepsilon (1+ r \varepsilon)^2 \leq 1/2 \Longrightarrow  r  L_{DF} (L_F)^2 \leq 1/2.
					\] 
					Thus, the conditions of Theorem \ref{ambientStability}
					are satisfied. Next, we reformulate the inequality \eqref{eq:BoundHausdorffDistance}
					of Theorem \ref{ambientStability}.
					The expression $E$ under the {square root} {at the right hand side} of {this} inequality 
					is:
					\begin{equation}\label{eq:E_expression}
						E = 1 +  (L_{F})^6  \left ( 1+4 r L_{DF} (L_{F})^{2} \right)^2   - 2    (L_{F})^3  \left( 1+4r L_{DF} (L_{F})^{2} \right)   \sqrt{1- (\varepsilon_2)^2}. 
					\end{equation}
					By replacing $L_{F}$ by $1+ L_{\varphi}$ in $E$, 
					the constants, as well as the
					degree-one terms in $L_{\varphi}$, ${r} L_{DF}$, and $\varepsilon_2$, cancel out. More precisely, 
					\begin{equation}\label{eq:ExpressionUnderRadical}
						E= 16 r^2  L_{DF}^2 + r^2 \varepsilon_2^2 + 24 r  L_{\varphi} L_{DF}  + 9 L_{\varphi}^2  + \bigO( |({r} L_{DF},L_{\varphi},\varepsilon_2 )|^3 ).
					\end{equation}

					Finally, by substituting  inequalities \eqref{eq:RelationEpsilonLipFDF}
					into \eqref{eq:ExpressionUnderRadical}, we obtain
					\[
					E \leq 50 r^2 \varepsilon^2 + \bigO\left(r^3 \varepsilon^{3} \right),
					\]
					and
					\[
					d_H (\ax(\Su), \ax ( F(\Su ) )) \leq \left( 1 + \sqrt{50}\right) r^2 \varepsilon 
					+ {\bigO\left(r^3\varepsilon^2 \right). }
					\]
				\end{proof}
				
				\begin{remark} \label{remark20}
					{We can also bound the expression $E$ from equation~\eqref{eq:E_expression} explicitly. Indeed, if $0\leq x \leq 1/4$ then $-\sqrt{1-x^2} \leq -1+ 4 (4-\sqrt{15})x^2$, and}
						\begin{align}
							E\leq  
							& 50.02  (r \varepsilon)^2 + 329.11  (r \varepsilon)^3 + 
							1126.37  (r \varepsilon)^4 + 2487.66  (r \varepsilon)^5 + 
							3841.65  (r \varepsilon)^6 
							\nonumber
							\\ &
							+ 4276.32  (r \varepsilon)^7 
							+ 3428.06  (r \varepsilon)^8 + 1928  (r \varepsilon)^9 + 
							720  (r \varepsilon)^{10} + 160  (r \varepsilon)^{11} + 
							16  (r \varepsilon)^{12},
							\label{Bound:E}
						\end{align}
						where we still assume that $r \varepsilon  \leq 1/4$.
				\end{remark}

				\begin{remark}\label{rem:Opt} 
					Observe 
					that the bound $ \bigO\left( r^2 \varepsilon \right)$ is consistent with a scaling 
					by factor $\lambda$: $\Su \mapsto \lambda \, \Su$, $F(\cdot) \mapsto \lambda \, F(\cdot/\lambda)$. Under such a scaling, 
					the radius $r$ is multiplied by $\lambda$, while the Lipschitz constant $\operatorname{Lip}(D\varphi)$ --- and therefore $\varepsilon$ --- is divided by $\lambda$. Furthermore, the Hausdorff distance $d_H (\ax(\Su), \ax ( F(\Su ) ))$
					increases by a factor $\lambda$. By considering a diffeomorphism that translates the set $\Su \setminus S(r)$ while keeping the bounding sphere $S(r)$ fixed,
					we see that this bound is 
					asymptotically optimal. 
				\end{remark}

				\section{Conclusion and future work}
				We proved the Hausdorff stability of the medial axis of a closed 
				set without any further assumption on it (as explained in Remark \ref{rem:BoundingSphereNoProblem}, the existence of the bounding sphere serves to formulate the main result in a clean way). 
				
				With regard to applications, our result is the first step towards providing a provably correct image recognition in the context of numerous scientific disciplines, and in particular of astrophysics. The next step is to produce physics-informed models for the medial axis as occurring in astronomical data. 
				
				On the mathematical side, we conclude with a conjecture generalizing our result to compact Riemannian manifolds with bounded curvature. 
				\begin{conjecture}
					Let $\M$ be a compact Riemannian manifold with bounded sectional curvature\footnote{See \cite{Berger} for definitions and a very pedagogical overview of the properties of these manifolds.} and $\Su$ a closed subset of $\M$. Then the medial axis (also called cut locus~\cite{kapovitch2021remarks}) of $\Su$ in $\M$ is Lipschitz stable under diffeomorphisms of $\M$. 
				\end{conjecture}

					\phantomsection
					\addcontentsline{toc}{section}{Bibliography}
					\bibliography{geomrefs}

					\appendix
					\section{Proofs of the Claims}
					\begin{proof}[Proof of Claim~\ref{claim:adjoint_operator}] 
						Since any invertible matrix $A$ satisfies $(A^t)^{-1} = (A^{-1})^t$, one has:
						\begin{align*}
							w \in D_p F (u^\perp) &\iff \langle D_pF^{-1}(w) , u \rangle =0 \\
							&\iff  \langle w , ( D_pF^{-1})^t u \rangle =0 \\
							&\iff  \langle w ,u' \rangle =0 \\
							&\iff w  \in  u'^\perp,
						\end{align*}
						and thus 
						\begin{equation}\label{eq:uPrimeOrthigonalImageT}
							D_p F (u^\perp) =  u'^\perp. 
						\end{equation}
						In other words, we have shown that $u'$ is orthogonal to $D_p F (u^\perp) = D_p F (T)$. 
						
						Because   
						\[
						\langle D_p F (u), (D_p F^{-1})^{t} (u) \rangle 
						= \langle D_p F^{-1}  (D_p F (u)),  u \rangle = \langle  u,  u \rangle > 0, \]
						we deduce that $\langle  D_pF (u), u'\rangle >0$. This is in turn equivalent to $u'$
						pointing towards the interior of $F(B(c,\rho))$.
					\end{proof} 
					
					\begin{proof}[Proof of Claim~\ref{lemma:BoundAngleUUPrime}]
						We first show that $\angle u, u'  < \pi/2$. 
						Indeed, define the vector $w$ as 
						\[w = (D_p F^t)^{-1} (u),
						\]
						that is, the vector satisfying $u = D_p F^t (w)$.
						Then $u' = \frac{w}{\abs{w}}$ (see equation~\eqref{eq:Def_u_prime}
						), and
						\begin{align}
							| w | \langle u, u' \rangle 
							&= \langle u, w \rangle  = \langle D_p F^t w, w \rangle 
							\nonumber 
							\\
							& = \langle   w, D_p F w \rangle
							= \abs{w}^2 + \langle   w, ( D_p F - \operatorname{Id} ) w \rangle 
							\nonumber
							\\
							&\geq \abs{w}^2 - \abs{w}^2 \|  D_p F - \operatorname{Id} \|
							\nonumber
							\\
							&> 0. 
							\tag{because, by assumption, $\|DF_p - \operatorname{Id}\| < 1$} 
						\end{align}
						Thus, $\langle u, u' \rangle > 0$, and therefore $\angle u, u' < \pi/2$.
						
						\begin{figure}[h!]
							\centering
							\includegraphics[width=.40\textwidth]{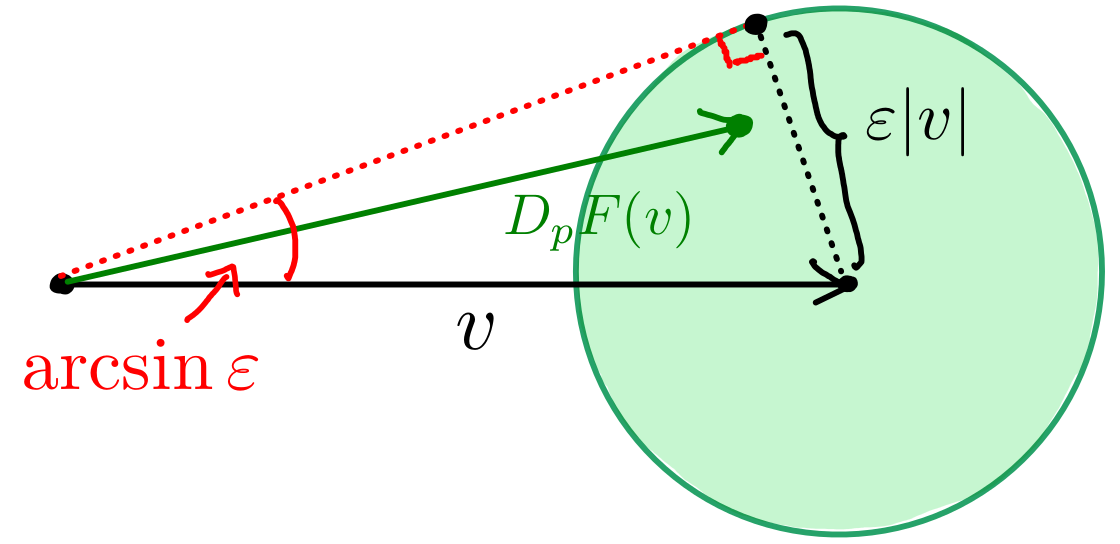}
							\caption{Since $ \| v - D_p F (v) \|\leq \varepsilon\abs{v}$, the vector $D_p F (v)$ lies in the green ball $B(v, \varepsilon\abs{v})$. Since $\varepsilon<1$, the angle between $v$ and $D_p F (v)$ is upper-bounded by $\arcsin\varepsilon<\pi/2$.
							}
							\label{fig:TrivialAngleBound}
						\end{figure}
						
						Furthermore, consider a vector $v\in u^\perp$. 
						Since $ \| v - D_p F (v) \|  \leq \|D_p F - \operatorname{Id}\| |v| \leq \varepsilon |v|$, the angle between $v$ and $D_p F (v)$ is upper-bounded by $\arcsin\varepsilon<\pi/2$, as illustrated in Figure~\ref{fig:TrivialAngleBound}. This yields a bound on the angle between the tangent spaces $u^\perp$ and $D_p F ( u^\perp )$:
						\begin{equation} \label{eq:AngleUperpAndItsImage}
							\sin  \angle  u^\perp, D_p F ( u^\perp ) 
							= \sin  \sup_{v \in u^\perp, w \in  D_p F ( u^\perp ) } \angle v, w \: \leq  \varepsilon.
						\end{equation}
						Using \eqref{eq:uPrimeOrthigonalImageT} and \eqref{eq:AngleUperpAndItsImage} we deduce that:
						\[
						\sin  \angle u, u' = \sin  \angle  u^\perp, u'^\perp \leq \varepsilon.
						\]
						Finally, since  $\angle u, u'  < \pi/2$, $\cos \angle u, u'\geq \sqrt{1-\varepsilon^2}$. This concludes the proof.
					\end{proof}
					
					\section{The calculation for \texorpdfstring{\eqref{Bound:E}}{the bound on $E$}}
					
					As mentioned, using that $-\sqrt{1-x^2} \leq -1+ 4 (4-\sqrt{15})x^2$ if $0\leq x \leq 1/4$, we find that $E$ is upper bounded by 
					\begin{align}
						E\leq & 1 +  (1+r \varepsilon)^6  \left ( 1+4 r \varepsilon (1+r \varepsilon)^{2} \right)^2   - 2    (1+r \varepsilon)^3  \left( 1+4r \varepsilon  (1+r \varepsilon)^{2} \right)   
						\nonumber 
						\\&+  2    (1+r \varepsilon)^3  \left( 1+4r \varepsilon (1+r \varepsilon)^{2} \right)  4 (4-\sqrt{15}) (r \varepsilon)^{2}
						\nonumber
						\\ =& 
						49 (r \varepsilon)^2
						+322 (r \varepsilon)^3
						+1103 (r \varepsilon) ^4
						+2446 (r \varepsilon) ^5
						+3801 (r \varepsilon) ^6
						+4256 (r \varepsilon) ^7
						\nonumber \\ &
						+3424 (r \varepsilon)^8
						+1928 (r \varepsilon)^9
						+720 (r \varepsilon)^{10}
						+160 (r \varepsilon)^{11}
						+16 (r \varepsilon)^{12}
						\nonumber \\ &
						+
						\left(4-\sqrt{15}\right)  \Big(  8 (r \varepsilon) ^2 
						+56  (r \varepsilon)^3 
						+184 (r \varepsilon) ^4
						+328  (r \varepsilon) ^5
						+320 (r \varepsilon) ^6
						+160  (r \varepsilon) ^7
						\nonumber \\ & \phantom{\left(4-\sqrt{15}\right)} 
						+32  (r \varepsilon)^8\Big) 
						\nonumber
						\\ =& 50.02  (r \varepsilon)^2 + 329.11  (r \varepsilon)^3 + 
						1126.37  (r \varepsilon)^4 + 2487.66  (r \varepsilon)^5 + 
						3841.65  (r \varepsilon)^6 
						\nonumber
						\\ &
						+ 4276.32  (r \varepsilon)^7 
						+ 3428.06  (r \varepsilon)^8 + 1928  (r \varepsilon)^9 + 
						720  (r \varepsilon)^{10} + 160  (r \varepsilon)^{11} + 
						16  (r \varepsilon)^{12},
						\tag{\ref{Bound:E}}
					\end{align}
					where we still assume that $r \varepsilon  \leq 1/4$.

					\section{Federer's tubular neighbourhood lemma}
					We recall:
					\begin{lemma}[Federer's tubular neighbourhood lemma, Theorem 4.8 (12) of \cite{Federer}]\label{Fed4.8.12}
						Let $p \in \Su$ and $\lfs (p)> 0$. The generalized normal space to $\Su$ at $p$ is characterized by the following property: For any 
						$\rho\in\R$ satisfying $0< \rho <\lfs (p) $,
						\begin{align}
							\Nor ( p,\Su) = \{ \lambda v\in \R^d \mid \lambda\geq 0 , | v| = \rho ,  \pi_\Su (p+v) =\{p\}\} .
							\nonumber
						\end{align}
						In particular, $\Nor ( p,\Su)$ is a convex cone.
						The generalized tangent space $\Tan(p,\Su)$ is the convex cone dual to $\Nor(p,\Su)$. 
					\end{lemma}

				\end{document}